\documentclass[reqno,a4paper]{amsart}

\usepackage{amsfonts,amssymb}

\DeclareMathOperator{\minor}{minor}
\DeclareMathOperator{\const}{const}

\theoremstyle{plain}

\newtheorem{theorem}{Theorem}
\newtheorem{lemma}{Lemma}

\theoremstyle{remark}

\newtheorem{remark}{Remark}

\begin{document}

\title[Matrix solution to pentagon with anticommuting variables]{A matrix solution to pentagon equation with anticommuting variables}

\author{S.~I.~Bel'kov, I.~G.~Korepanov}

\begin{abstract}
We construct a solution to pentagon equation with anticommuting variables living on two-dimensional faces of tetrahedra. In this solution, matrix coordinates are ascribed to tetrahedron vertices. As matrix multiplication is noncommutative, this provides a ``more quantum'' topological field theory than in our previous works.
\end{abstract}

\keywords{Pentagon equation, topological quantum field theory, algebraic complex, torsion}

\maketitle

\section{Introduction}

Pentagon equation deals with a Pachner move $2\to 3$ and is a fundamental constituent of many topological quantum field theories (TQFT's) for three-dimensional manifolds. Pachner moves are elementary rebuildings of a manifold triangulation whose importance is due to the theorem of Pachner~\cite{pachner,lickorish}: it states (in particular) that, for a given three-dimensional manifold, any triangulation can be obtained from any other by a finite sequence of Pachner moves of just four kinds: $2\leftrightarrow 3$ and $1\leftrightarrow 4$.

Here ``$2\leftrightarrow 3$'' means the following. Let there be in the triangulation, among others, two tetrahedra having a common two-dimensional face. We denote them $1234$ and~$5123$, where $1,\dots,5$ are their vertices, so $123$ is their common 2-face. The $2\to 3$ move, by definition, replaces these two tetrahedra with three tetrahedra $1254$, $2354$ and~$3154$, occupying the same domain in the manifold. The $3\to 2$ move is the inverse operation.

As for the move $1\to 4$, it inserts a new vertex~$5$ into a tetrahedron~$1234$ and replaces it with four tetrahedra $1235$, $1425$, $1345$, and~$3245$. Move $4\to 1$ is again the inverse operation. Usually, however, the central role is played by moves $2\leftrightarrow 3$: if we have managed to do something meaningful for them, then it so happens that our construction works for moves $1\leftrightarrow 4$ ``automatically'' --- and we will meet with this exactly situation below when proving Theorem~\ref{th:i} of the present work.

As Pachner moves relate \emph{any} two triangulations of a given manifold, a quantity invariant under all Pachner moves does not depend on a specific triangulation and is thus a \emph{manifold invariant}. To be exact, this applies to piecewise-linear (PL) manifolds. In three dimensions, however, the piecewise-linear category coincides with the topological category~\cite{PL_and_Top}, so we get a topological invariant as well.

The specific mathematical sense of ``pentagon equation'' can be different in different scientific papers, the situation that is well-known also, e.g., for Yang--Baxter equation: in both cases, an equation ``with variables on the edges'', or on the 2-faces, etc., can be considered, as well as functional (or ``set-theoretic'') version of equation, and so on. We call \emph{solution to pentagon equation} any algebraic relation that can be reasonably said to correspond to a $2\to 3$ Pachner move and from which one can expect that manifold invariants can be constructed. In this paper, our solution to pentagon equation is formulated in terms of Grassmann--Berezin anticommuting variables and is, in this respect, similar to the solution in~\cite{bkm}, where it appeared as the first step to constructing a simple TQFT\footnote{Note also a long way form formula~(5) in paper~\cite{k-jnmp-2001}, which gave the origin to our research on pentagon equations, to a TQFT in papers~\cite{tqft,tqft2}.} related to group~$\mathrm{PSL}(2,\mathbb C)$.

Note that our TQFT's are \emph{finite-dimensional} in the sense that they involve no functional (infinite-dimensional) integrals: they deal with finite triangulations and ascribe to them a finite number of quantities.

The distinguishing feature of the present paper is that the parameters of the theory --- so-called ``coordinates'' ascribed to triangulation vertices --- are matrices, and the noncommutative matrix multiplication plays essential role in our pentagon equation here. This may be important, because, putting it a bit informally, any new noncommutativity makes our theories ``more quantum'' and thus removes ``classical'' features which might be present in our previous TQFT's.

There are some considerations showing that this ``more quantum'' character will manifest itself properly only in the context involving nontrivial --- and even non-abelian --- representation of the manifold's fundamental group, like it is described in papers~\cite{M1,M2}. The aim of the present short paper is, however, just to construct the solution to pentagon equation and show that it works also for moves~$1\leftrightarrow 4$, so we leave those ``quantum'' calculations for further work.

Below, we begin in Section~\ref{s:scalar} with presenting our constructions in the scalar (matrices~$1\times 1$) case, which already gives a new and elegant solution to pentagon equation. The generalization to matrices~$n\times n$ is not so straightforward, it arises from some specific algebraic complexes, introduced in Section~\ref{s:complexes}. In Section~\ref{s:genfun}, we establish the connection bewteen matrices of linear mappings in algebraic complexes and expressions in anticommuting variables. The actual solution to pentagon equation with matrices~$n\times n$ is presented in Section~\ref{s:mpen}. Then we construct in Section~\ref{s:inv} the (simplest version of) related manifold invariants; here moves $1\leftrightarrow 4$ come into play. Finally, in Section~\ref{s:d} we discuss some miracles encountered in previous sections.

\section{Solution to pentagon equation: the scalar case}
\label{s:scalar}

\subsection{Grassmann algebras and Berezin integral}

Recall~\cite{B} that \emph{Grassmann algebra} over field~$\mathbb F=\mathbb R$ or~$\mathbb C$ is an associative algebra with unity, having generators~$a_i$ and relations
$$
a_i a_j = -a_j a_i, \quad \textrm{including} \quad a_i^2 =0.
$$
Thus, any element of a Grassmann algebra is a polynomial of degree $\le 1$ in each~$a_i$.

The \emph{Berezin integral}~\cite{B} is an $\mathbb F$-linear operator in a Grassmann algebra defined by equalities
\begin{equation}
\int da_i =0, \quad \int a_i\, da_i =1, \quad \int gh\, da_i = g \int h\, da_i,
\label{iB}
\end{equation}
if $g$ does not depend on~$a_i$ (that is, generator $a_i$ does not enter the expression for~$g$); multiple integral is understood as iterated one.

\subsection{Solution to pentagon equation}

Consider a tetrahedron with vertices $i_1,i_2,\allowbreak i_3,i_4$, or simply tetrahedron~$i_1i_2i_3i_4$. We introduce a complex parameter~$\zeta_i$ for every vertex~$i$, called its ``coordinate''. These parameters are arbitrary, with the only condition that any two different vertices~$i\ne j$ have different coordinates~$\zeta_i\ne\zeta_j$. We will also use the notation
$$
\zeta_{ij} \stackrel{\rm def}{=} \zeta_i - \zeta_j.
$$

Then, we put in correspondence to any (unoriented) 2-face~$ijk$ a Grassmann generator~$a_{ijk}\; (=a_{ikj}=\dots=a_{kji})$, and to the tetrahedron~$i_1i_2i_3i_4$ --- its \emph{weight}
\begin{multline}
\mathbf f_{i_1i_2i_3i_4} \stackrel{\rm def}{=} \zeta_{i_1i_2}a_{i_1i_2i_3}a_{i_1i_2i_4} - \zeta_{i_1i_3}a_{i_1i_3i_2}a_{i_1i_3i_4} + \zeta_{i_1i_4}a_{i_1i_4i_2}a_{i_1i_4i_3} \\
+ \zeta_{i_2i_3}a_{i_2i_3i_1}a_{i_2i_3i_4} - \zeta_{i_2i_4}a_{i_2i_4i_1}a_{i_2i_4i_3} + \zeta_{i_3i_4}a_{i_3i_4i_1}a_{i_3i_4i_2} .
\label{fbf}
\end{multline}
Note that, in each summand in~\eqref{fbf}, the~$\zeta$ belongs to an edge, while the two~$a$'s --- to the two adjacent faces.

The weight~\eqref{fbf} is the simplest example of a \emph{generating function} of invariants of manifold with triangulated boundary; the invariants are the coefficients at the products of anticommuting variables. This makes, of course, little sense when the manifold is just one tetrahedron, but becomes nontrivial already in the case of clusters of two and three tetrahedra in Theorem~\ref{th:5f} below.

\begin{remark}
Expression~\eqref{fbf} changes its sign under an odd permutations of indices~$i_1,\allowbreak i_2,\allowbreak i_3,i_4$, i.e., it belongs to an \emph{oriented} tetrahedron~$i_1i_2i_3i_4$; orientation is understood as an ordering of tetrahedron vertices up to even permutations. It will be convenient for us, however, mostly to ignore the orientations in this paper and simply write the vertices in the increasing order of their numbers, like in the following Theorem~\ref{th:5f}.
\end{remark}

\begin{theorem}
\label{th:5f}
The function $\mathbf f_{i_1i_2i_3i_4}$ defined by~\eqref{fbf} satisfies the following pentagon equation (dealing with two tetrahedra $1234$ and~$1235$ in its l.h.s.\ and three tetrahedra $1245$, $2345$ and~$1345$ in its r.h.s.):
\begin{equation}
\int \mathbf f_{1234} \mathbf f_{1235} \, da_{123} = \frac{1}{\zeta_{45}} \iiint \mathbf f_{1245} \mathbf f_{2345} \mathbf f_{1345} \, da_{145}\, da_{245}\, da_{345} .
\label{5f}
\end{equation}
\end{theorem}

\begin{proof}
Formula~\eqref{5f} can be proved, e.g., by a computer calculation.
\end{proof}

\begin{remark}
The integration in both sides of~\eqref{5f} goes in the Grassmann variables living at the inner 2-faces of the corresponding cluster of two or three tetrahedra. The special role of edge~$45$ in~\eqref{5f}, manifested in the factor $1/\zeta_{45}$, corresponds to the fact that~$45$ is the only inner edge among the ten edges of the r.h.s.\ tetrahedra.
\end{remark}

\subsection{A tentative state-sum invariant in the scalar case}
\label{ss:renorm}

If there is a triangulated oriented manifold~$M$ with boundary, then one can construct the following function of anticommuting variables~$a_{ijk}$ living on \emph{boundary} faces (and parameters~$\zeta_i$ in vertices):
\begin{equation}
\frac{1}{\prod\nolimits'\zeta_{ij}}\, \int\!\dots\!\int \prod \mathbf f_{klmn} \prod\nolimits' da_{ijk},
\label{s}
\end{equation}
where the product $\prod\nolimits'\zeta_{ij}$ goes over all \emph{inner} edges~$ij$, the product $\prod \mathbf f_{klmn}$ --- over all tetrahedra~$klmn$, and $\prod\nolimits' da_{ijk}$ --- over all inner faces. The expression~\eqref{s} is determined up to an overall sign which may change if with change the order of the th vertices (and/or tetrahedra, differentials, etc.). It is a quite obvious consequence from Theorem~\ref{th:5f} and Remark following it that \eqref{s} is at least invariant under all Pachner moves~$2\leftrightarrow 3$ not changing the boundary.

It turns out that~\eqref{s} is already, in some cases, a working multicomponent (that is, incorporating many coefficients at various monomials in anticommuting variables) invariant. It can be called a \emph{state sum} for manifold~$M$; from a physical viewpoint, the anticommuting variables mean that this state sum is of \emph{fermionic} nature. It can be shown, however, that there are two difficulties with direct application of~\eqref{s}:
\begin{itemize}
\item if the triangulation has at least one inner (not boundary) vertex, \eqref{s} yields zero,
\item if the boundary of a connected manifold has more than one connected component, \eqref{s} also yields zero.
\end{itemize}

These are two reasons for introducing more powerful technique for obtaining manifold invariants. The third reason is that the noncommutative (matrix) generalization of weight~\eqref{fbf} is neither straightforward nor obvious. It turns out that these problems are solved by introducing new variables, united in an \emph{algebraic (chain) complex}.

\section{Algebraic complexes with matrix ``coordinates''}
\label{s:complexes}

\subsection{Explicit formulas}
\label{ss:c}

We consider a triangulated three-dimensional compact oriented connected manifold~$M$ with \emph{one-component}\footnote{The case where $\partial M$ has exactly one connected component is the easiest technically and seems to be enough for the present short paper. The complications arising when $\partial M$ is allowed to have arbitrary number $0,1,2,\dots$ of components are not very big, and such situation for a similar construction has been considered in~\cite{bkm}.} boundary~$\partial M$. We will eventually present, below in Section~\ref{s:inv}, a set of invariants, constructed for the given \emph{boundary} triangulation and depending on $n\times n$ complex matrices~$\zeta_i$ assigned to each boundary vertex~$i$; every individual invariant from the set corresponds to a certain coloring of boundary faces. Here coloring means choosing some set~$\mathcal C$ of certain differentials, this will be explained soon after formula~\eqref{c}.

We present (a simple version of) our construction of algebraic complexes, providing, in particular, the matrix generalization of weight~\eqref{fbf}. In this subsection, we present the formulas defining our algebraic complexes in the explicit form: essentially, as a sequence of three matrices~$f_2,f_3,f_4$. These formulas are well suited for computer calculations, although their form can hardly explain how they were found and for what reason our sequence~\eqref{c} of vector spaces and linear mappings is indeed an algebraic complex. This is explained in the next Subsection~\ref{ss:macro}.

We denote by~$N_k$, $k=0,1,2,3$, the number of $k$-simplexes in the triangulation, and by $N'_k$ --- the number of inner $k$-simplexes.

Then we number all vertices, in some arbitrary order, by numbers $i=1,\dots,N_0$.

Our invariants come out from algebraic (chain) complexes of the following form:
\begin{equation}
0 \longrightarrow \mathbb C^{n\cdot N'_0} \stackrel{f_2}{\longrightarrow} \mathbb C_{\mathcal C}^{2n\cdot N_3} \stackrel{f_3}{\longrightarrow} \mathbb C^{2n\cdot N_3} \stackrel{f_4}{\longrightarrow} \mathbb C^{n\cdot N'_0} \longrightarrow 0.
\label{c}
\end{equation}
We consider each vector space in~\eqref{c} as consisting of column vectors of the height equal to the exponent at~$\mathbb C$. All vector spaces have thus natural \emph{distinguished bases} consisting of vectors with one coordinate unity and all other zero\footnote{This is important when we are dealing with subject related to Reidemeister-style torsions. These will appear below in Section~\ref{s:inv}.}; thus we can, and do, identify them with their matrices.

\begin{remark}
The first nonzero mapping in~\eqref{c} is denoted~$f_2$, and not~$f_1$, in order to match our notations here with other papers, where similar but longer complexes appear, including two more mappings called $f_1$ and~$f_5$. See also the next Subsection~\ref{ss:macro}.
\end{remark}

A column vector --- element of the first (from the left) space~$\mathbb C^{n\cdot N'_0}$ is made, by definition, of $N'_0$ vectors~$dz_i$ corresponding to each inner vertex~$i$ and each having $n$ components.

The next space, $\mathbb C_{\mathcal C}^{2n\cdot N_3}$, requires a longer explanation. Let there be a 2-face~$ijk$, with $i<j<k$. To such a face corresponds, by definition, a column vector~$d\varphi_{ijk}$ of height~$n$. An element of the vector space in question consists, by definition, of all elements of all~$d\varphi_{ijk}$ corresponding to $N'_2$ \emph{inner} faces, and of some set~$\mathcal C$ of cardinality $\#\mathcal C=n\cdot(2N_3-N'_2)$ of components of~$d\varphi_{ijk}$ corresponding to boundary faces~$ijk$.

We would like, however, to define some more quantities for our further needs. We denote by~$b$ any ordered triple~$ijk$ of triangulation vertices corresponding to some 2-face in the triangulation. Here ``ordered'' means that we take them in this exact order: $i,j,k$, ignoring which of numbers~$i$, $j$ and~$k$ is smaller or greater. Now, if $i<j<k$, we set by definition
\begin{equation}
d\varphi_{i,b}\stackrel{\rm def}{=}d\varphi_b.
\label{phi}
\end{equation}
Then we define $d\varphi_{i,b}$ for any pair $i,b$ with $i\in b$ by the following conditions:
\begin{itemize}
\item if $b_2$ is obtained from $b_1$ by an odd permutation of $i,j,k$, then $d\varphi_{i,b_2}=-d\varphi_{i,b_1}$ (thus, for an \emph{even} permutation, the two $d\varphi_{i,b}$ are of course equal),
\item the following relations hold:
\begin{eqnarray}
d\varphi_{i,b}+d\varphi_{j,b}+d\varphi_{k,b}&=&0, \label{phi0} \\
\zeta_i\, d\varphi_{i,b}+\zeta_j\, d\varphi_{j,b}+\zeta_k\, d\varphi_{k,b}&=&0. \label{phi1}
\end{eqnarray}
\end{itemize}

We now pass on to the following space, also $\mathbb C^{2n\cdot N_3}$. Let there be a tetrahedron~$ijkl$, with $i<j<k<l$, also denoted by a single letter~$a$. To such a tetrahedron correspond, by definition, \emph{two} column vectors $d\psi_{i,a}$ and~$d\psi_{j,a}$, each of height~$n$. An element of the vector space consists, by definition, of all such column vectors together.

We would like, however, to define again some more quantities, namely, $d\psi_{i,a}$ for any vertex~$i$ and any tetrahedron~$a\ni i$, \emph{regardless} of condition $i<j<k<l$. We do it in analogy with what we have done for faces, by imposing the following conditions:
\begin{itemize}
\item if $a_2$ is obtained from $a_1$ by an odd permutation of $i,j,k$, then $d\psi_{i,a_2}=-d\psi_{i,a_1}$,
\item the following relations hold:
\begin{eqnarray}
d\psi_{i,a}+d\psi_{j,a}+d\psi_{k,a}+d\psi_{l,a}&=&0, \label{psi0} \\
\zeta_i\, d\psi_{i,a}+\zeta_j\, d\psi_{j,a}+\zeta_k\, d\psi_{k,a}+\zeta_l\, d\psi_{l,a}&=&0. \label{psi1}
\end{eqnarray}
\end{itemize}

Finally, an element of the last space~$\mathbb C^{n\cdot N'_0}$ is similar to that in the first space: it consists of $N'_0$ vectors~$d\chi_i$ corresponding to each inner vertex~$i$ and each having $n$ components.

We define linear mappings~$f_2$, $f_3$ and~$f_4$ as follows.

\begin{itemize}
\item $f_2$, by definition, makes the following $d\varphi_{ijk}$ from given~$dz_i$:
\begin{equation}
d\varphi_{ijk}=(\zeta_i-\zeta_j)^{-1}(dz_i-dz_j)-(\zeta_i-\zeta_k)^{-1}(dz_i-dz_k).
\label{f2}
\end{equation}
\item $f_3$, by definition, makes the following $d\psi_{i,a}$ and $d\psi_{j,a}$, where $a=ijkl$, $i<j<k<l$, from given~$d\varphi$'s:
\begin{equation}
d\psi_{i,a}=d\varphi_{i,ijk}+d\varphi_{i,ikl}+d\varphi_{i,ilj}, \qquad
d\psi_{j,a}=d\varphi_{j,ijk}+d\varphi_{j,ilj}+d\varphi_{j,jlk},
\label{f3}
\end{equation}
where the $d\varphi$'s in the r.h.s.\ are of course calculated using \eqref{phi}, \eqref{phi0} and~\eqref{phi1}.
\item $f_4$, by definition, makes the following $d\chi_{i}$ from given~$d\varphi_{ijk}$:
\begin{equation}
d\chi_{i}=\sum_a d\psi_{i,a},
\label{f4}
\end{equation}
with the sum taken over all tetrahedra~$a$ surrounding the given vertex~$i$ and taken all with \emph{positive} orientation; the $d\psi$'s in~\eqref{f4} are calculated , if necessary, using formulas \eqref{psi0} and~\eqref{psi1}.
\end{itemize}

\begin{remark}
\label{r:c}
Matrix~$f_3$ depends thus on the chosen set~$\mathcal C$ of components of~$d\varphi_{ijk}$ corresponding to boundary faces~$ijk$, as explained above. All such matrices are, obviously, submatrices of matrix~$f_3^{\rm full}$ incorporating all rows corresponding to \emph{all} components of~$d\varphi_{ijk}$. Matrix~$f_3^{\rm full}$ acts thus from $\mathbb C^{n\cdot N_2}$ to~$\mathbb C^{2n\cdot N_3}$, we will make use of it below in Section~\ref{s:mpen}.
\end{remark}

\begin{theorem}
\label{th:c}
The sequence~\eqref{c} of vector spaces and linear mappings is indeed an algebraic complex, i.e.,
\begin{equation}
f_3\circ f_2=0\quad\hbox{and}\quad f_4\circ f_3=0.
\label{ff}
\end{equation}
\end{theorem}

\begin{proof}
Theorem~\ref{th:c} can be proved by a direct calculation.
\end{proof}

\subsection{The mathematical origins of complex~(\ref{c})}
\label{ss:macro}

The proof of Theorem~\ref{th:c} by means of direct calculation does not make clear the mathematical reasons ensuring that~\eqref{c} is a complex. So, in this subsection we briefly explain the mathematical origins\footnote{In this connection, see Acknowledgements in the end of the paper.} of complex~\eqref{c}. Namely, the linear mappings $f_2$, $f_3$ and~$f_4$ appear as \emph{differentials}~$dF_2$, $dF_3$ and~$dF_4$ (with some modifications/refinements if necessary) of some mappings~$F_i$ forming a sort of ``nonlinear complex'' in the sense that
\begin{equation}
F_{i+1}\circ F_i = \const .
\label{FF}
\end{equation}
Then it obviously follows from~\eqref{FF} that
$$
dF_{i+1}\, dF_i = 0 ,
$$
to be compared with~\eqref{ff}.

There are actually five mappings~$F_i$: $i=1,\dots,5$. But, as Theorem~\ref{th:c} is already proved, and here we just want to give an idea of where the formulas in Subsection~\ref{ss:c} come from, we restrict ourselves to presenting first three of~$F_i$, leaving $F_4$ and~$F_5$ as an exercise for an interested reader.

Let there be a triangulated three-dimensional manifold, with fixed $n\times n$ matrices~$\zeta_i$ assigned to each triangulation vertex~$i=1,\dots,N_0$.

\begin{itemize}

\item By definition, $F_1$ takes a pair $(a,b)$ of $n\times n$ matrices to
\begin{equation}
z_i=\zeta_i a+b
\label{F1}
\end{equation}
for each vertex~$i$ in the triangulation. Thus, $a$ and~$b$ are, essentially, parameters of the \emph{group of affine transformations} of $n\times n$ matrix algebra. The corresponding tangent mapping~$f_1=dF_1$ does not appear in our complex~\eqref{c}, but will appear in its more general versions.

\item Mapping $F_2$ takes, by definition, matrices~$z_1$ in vertices to matrices
\begin{equation}
\varphi_{ijk}=\zeta_{ij}^{-1} z_{ij} z_{ik}^{-1} \zeta_{ik} 
\label{F2}
\end{equation}
for all two-faces $ijk$ with $i<j<k$. Here and below we use notations
\begin{equation*}
\zeta_{ij}\stackrel{\rm def}{=}\zeta_i-\zeta_j,\quad z_{ij}\stackrel{\rm def}{=}z_i-z_j, \hbox{ etc.}
\end{equation*}
Note that for ``initial'' values $z_{\dots}=\zeta_{\dots}$, we have $\varphi=1$. Obviously, the first of equalities~\eqref{FF}, namely $F_2\circ F_1=\const $, holds. 

The linear mapping~$f_2$ in our specific complex~\eqref{c} is obtained, first, by differentiating formula~\eqref{FF} with respect to~$z_i$ for \emph{inner} vertices~$i$ at their ``initial'' values $z_i=\zeta_i$, while $z_i=\zeta_i$ for boundary~$i$ stay constant. This gives formula~\eqref{f2}, where both $dz_i$ and $d\varphi_{ijk}$ are, at this moment, $n\times n$ matrices. As, however, \eqref{f2} obviously operates with each column of~$dz_i$ and corresponding column of~$d\varphi_{ijk}$ separately, we then consider just one (e.g., first) column in both matrices, leaving for this column, a bit loosely, the same respective notation $dz_i$ or~$d\varphi_{ijk}$.

\item Mapping $F_3$, by definition, takes matrices~$\varphi_{ijk}$ to matrices~$\psi_{i,a}$ associated with every tetrahedron $a=ijkl$ and its vertex~$i$, according to the following formula:
\begin{equation}
\psi_{i,a}=\varphi_{ijk}\varphi_{ikl}\varphi_{ilj}. 
\label{F3}
\end{equation}
Here, the $\varphi$'s for \emph{any} order of their vertices are calculated according to formulas
\begin{eqnarray*}
&&\varphi_{ikj}=\varphi_{ijk}^{-1},\\
&& \zeta_{ij}^{-1}\zeta_{ik}-\varphi_{ijk}=-\zeta_{ji}^{-1}\zeta_{jk} \varphi_{kij}^{-1}, 
\label{e:varphi:permutations}
\end{eqnarray*}
which is in agreement with~\eqref{F2}.

Again, it is quite obvious that $F_3\circ F_2=\const $. Formula~\eqref{f3} is obtained from~\eqref{F3} by differentiating, again at the initial values $z_i=\zeta_i$, $\varphi_{ijk}=1$, and then taking single columns in place of matrix differentials, like we did it for matrix~$f_2$.

\end{itemize}

\section{Generating functions of Grassmann variables for rectangular matrices}
\label{s:genfun}

We now want to link matrices (having in mind mostly matrix~$f_3$ from Subsection~\ref{ss:c}) to functions of Grassmann variables.

Let $A$~be an arbitrary matrix whose entries are complex numbers or complex-valued expressions, with the only condition that the number of columns is not smaller than the number of rows\footnote{Note that here we have swapped the roles of rows and columns with respect to paper~\cite{bkm}!}.

With each column~$k$ of~$A$, we associate a Grassmann generator~$a_k$, while with the whole matrix~$A$ --- the \emph{generating function} defined as
\begin{equation}
\mathbf f_A = \sum_{\mathcal C} \det A|_{\mathcal C} \prod_{k\in \mathcal C} a_k,
\label{genfun}
\end{equation}
where $\mathcal C$ runs over all subsets of the set of columns of the cardinality equal to the number of rows; $A|_{\mathcal C}$ is the square submatrix of~$A$ containing all columns in~$\mathcal C$; the order of~$a_k$ in the product is the same as the order of columns in~$A|_{\mathcal C}$ (e.g., the most natural --- increasing --- order of~$k$'s in both).

\begin{lemma}
\label{l:mult}
Let $C$ be the vertical concatenation of matrices $A$ and~$B$ having the equal number of rows: 
$C=\left(\begin{matrix} A \\ B \end{matrix}\right)$.
Then
$$
\mathbf f_C = \mathbf f_A \mathbf f_B .
$$
\end{lemma}

\begin{proof}
The lemma easily follows from the expansion of the form
\begin{equation}
\minor C = \sum \pm \minor A \, \minor B,
\label{pm}
\end{equation}
known from linear algebra, for every minor of~$C$ having the full number of columns.
\end{proof}

Let there be now a subset~$\mathcal I$ in the set of all columns of~$A$. We call the columns in~$\mathcal I$ \emph{inner}, the rest of them --- \emph{outer}, and we define the \emph{generating function of matrix~$A$ with the set~$\mathcal I$ of inner columns} as
\begin{equation}
{}_{\mathcal I}\mathbf f_A = \sum_{\mathcal C \supset \mathcal I} \det\nolimits' A|_{\mathcal C} \prod_{k\in \mathcal C \setminus \mathcal I} a_k.
\label{genfunI}
\end{equation}
Here $\det\nolimits'$ means that, unlike in~\eqref{genfun}, we are changing the order of $A$'s columns in the following way: all inner columns are brought to the right of the matrix; the order of columns within the set~$\mathcal I$ and its complement is conserved; the order of~$a_k$'s in the product (where $k$ belongs to the mentioned complement) is the same as the order of columns~$k$.

\begin{lemma}
\label{l:int}
The generating function of matrix~$A$ with the set~$\mathcal I$ of inner columns is the following Berezin integral of the usual generating function:
\begin{equation}
{}_{\mathcal I}\mathbf f_A = \int\!\dots\!\int \mathbf f_A \prod_{l\in\mathcal I} da_l,
\label{f_inner}
\end{equation}
where the differentials are written in the same\footnote{We adopt the convention that the multiple Berezin integral is calculated following the rule $\iint f(a) g(b) \, da \, db = \int f(a)\, da \cdot \int g(b)\, db $. This convention seems most commonly accepted. Note that we were using a different covention in~\cite{bkm}, with differentials in a multiple integral written in the reverse order --- hence the difference between \eqref{f_inner} and \cite[formula~(52)]{bkm}.} order as rows in~$A$.
\end{lemma}

\begin{proof}
First, we note that only those terms in~$\mathbf f_A$ survive the integration in the r.h.s.\ of~\eqref{f_inner} which contain all the~$a_k$ for~$k\in \mathcal I$. We take the function~$\mathbf f_A$ as defined in~\eqref{genfun}, leave only the mentioned terms in it, and note that none of them is changed if we bring both the columns~$k$ in~$A$ for all~$k\in \mathcal I$ to the right of the matrix and the corresponding generators~$a_k$ to the right in the product\footnote{because any elementary permutation of columns brings a minus sign which cancels out with the minus brought by the corresponding permutation of~$a_k$'s}, neither changing the order within~$\mathcal I$ nor within its complement. Then, the integration in~\eqref{f_inner} just takes away the~$a_k$ for~$k\in \mathcal I$, as required.
\end{proof}

\section{Pentagon equation with matrix coordinates}
\label{s:mpen}

\subsection{The pentagon equation}
\label{ss:p}

It turns out that the matrix version of tetrahedron weight~\eqref{fbf}, satisfying the (matrix version of) pentagon equation, can be constructed as the generating function for matrix~$f_3^{\rm full}$ (see Remark~\ref{r:c}) corresponding to just one tetrahedron~$a=i_1i_2i_3i_4$ considered as a manifold with boundary\footnote{Note that, as it has no inner vertices, there are no matrices $f_2$ and $f_4$ in complex~\eqref{c} written for a single tetrahedron, in the sense that one of dimensions in both matrices is zero. The same applies to the l.h.s.\ and r.h.s.\ of Pachner move $2\to 3$ below in Lemma~\ref{l:23}.}. This matrix~$f_3^{\rm full}$ can be calculated using \eqref{f3}, \eqref{phi0} and~\eqref{phi1}, and reads:
\begin{equation}
f_3^{\rm full} =  \begin{pmatrix} \mathbf 1 & \mathbf{-1} & \mathbf 1 & \mathbf 0 \\
\zeta_{i_2i_3}^{-1}\zeta_{i_3i_1} & -\zeta_{i_2i_4}^{-1}\zeta_{i_4i_1} & \mathbf 0 & \mathbf{-1} \end{pmatrix} .
\label{f3t}
\end{equation}
Matrix~\eqref{f3t} is of course a block matrix, with $\mathbf 0$ and~$\mathbf 1$ meaning the $n\times n$ zero and unity matrices, respectively. The block rows of matrix~\eqref{f3t} correspond to differentials $d\psi_{i_1,a}$ and~$d\psi_{i_2,a}$, while the columns --- to $d\varphi_{i_1i_2i_3}$, $d\varphi_{i_1i_2i_4}$, $d\varphi_{i_1i_3i_4}$, and~$d\varphi_{i_2i_3i_4}$ (in the natural order in both cases). We denote~$\mathbf f_{i_1i_2i_3i_4}$ the generating function~\eqref{genfun} for matrix~\eqref{f3t}; we do not write it out here because, as computer calculation shows, it contains 60 nonzero monomials already in the case $n=2$.

To each $d\varphi_{ijk}$ corresponds thus a \emph{vector}~$a_{ijk}$ of $n$ anticommuting variables.

\begin{theorem}
\label{th:mpen}
The generating functions~$\mathbf f_{i_1i_2i_3i_4}$ for matrices~\eqref{f3t} satisfy the following pentagon equation:
\begin{equation}
\frac{\det\zeta_{23}}{\det\zeta_{34}\,\det\zeta_{35}} \underbrace{\int\!\dots\!\int}_{n} \mathbf f_{1234} \mathbf f_{1235} \,\mathcal Da_{123} = \frac{1}{\det\zeta_{45}} \underbrace{\int\!\dots\!\int}_{3n} \mathbf f_{1245} \mathbf f_{2345} \mathbf f_{1345} \,\mathcal Da_{145} \,\mathcal Da_{245} \,\mathcal Da_{345} ,
\label{pmatr}
\end{equation}
where $\mathcal Da_{i_1i_2i_3}$ means the product of all components of~$da_{i_1i_2i_3}$, taken in their natural order.
\end{theorem}

\begin{proof}
We first prove the following lemma.

\begin{lemma}
\label{l:23}
The $n$-fold and $3n$-fold integrals in the l.h.s.\ and r.h.s.\ of~\eqref{pmatr} are generating functions for matrices~$f_3^{\rm full}$ corresponding to the l.h.s.\ and r.h.s.\ of Pachner move $2\to 3$, respectively, considered as triangulated manifolds with boundary.
\end{lemma}

\begin{proof}
Both these matrices~$f_3^{\rm full}$ are vertical concatenations of matrices~\eqref{f3t} extended with necessary columns, corresponding to 2-faces absent from the given tetrahedron and filled with zeros. For instance, here is the matrix~$f_3^{\rm full}$ for the l.h.s.\ of move $2\to 3$:
\begin{equation}
(f_3^{\rm full})_{\rm l.h.s.} =
\begin{pmatrix} \mathbf 1 & \mathbf{-1} & \mathbf 0 & \mathbf 1 & \mathbf 0 & \mathbf 0 & \mathbf 0 \\
\zeta_{23}^{-1}\zeta_{31} & -\zeta_{24}^{-1}\zeta_{41} & \mathbf 0 & \mathbf 0 & \mathbf 0 & \mathbf{-1} & \mathbf 0 \\
\mathbf 1 & \mathbf 0 & \mathbf{-1} & \mathbf 0 & \mathbf 1 & \mathbf 0 & \mathbf 0 \\
\zeta_{23}^{-1}\zeta_{31} & \mathbf 0 & -\zeta_{25}^{-1}\zeta_{51} & \mathbf 0 & \mathbf 0 & \mathbf 0 & \mathbf{-1}  \end{pmatrix} .
\label{4*7}
\end{equation}

Like in \eqref{f3t}, every element of matrix~\eqref{4*7} is a matrix of sizes $n\times n$. The first block column in matrix~\eqref{4*7} corresponds to the ($n$ components of) differential~$d\varphi_{123}$ at the \emph{inner} face~$123$, the rest of columns --- to the following \emph{boundary} faces, from left to right: $d\varphi_{124}$, $d\varphi_{125}$, $d\varphi_{134}$, $d\varphi_{135}$, $d\varphi_{234}$, and~$d\varphi_{235}$. The rows correspond to $d\psi_{1,1234}$, $d\psi_{2,1234}$, $d\psi_{1,1235}$, and~$d\psi_{2,1235}$.

We do not write out here the matrix corresponding to the r.h.s. of the Pachner move. It is made in the same obvious manner and contains $6\times 9$ block entries.

The statement of the lemma follows now from Lemmas \ref{l:mult} and~\ref{l:int}.
\end{proof}

Now we continue with the proof of Theorem~\ref{th:mpen}. It remains to prove that the minors of matrix~$(f_3^{\rm full})_{\rm l.h.s.}$ and the similar $6\times 9$ block matrix~$(f_3^{\rm full})_{\rm r.h.s.}$, corresponding to the r.h.s.\ of the Pachner move $2\to 3$, are proportional with the same ratio as the integrals in both sides of~\eqref{pmatr}, provided these minors contain all the rows of the corresponding matrix, all the columns corresponding to~$d\varphi$'s at inner faces (or simply ``inner~$d\varphi$''), and the other columns in two minors correspond to the same~$d\varphi$'s at boundary faces (or simply ``boundary~$d\varphi$'').

Let $f_3$ denote, in the rest of this proof, any of $(f_3^{\rm full})_{\rm l.h.s.}$ and~$(f_3^{\rm full})_{\rm r.h.s.}$. Consider the following question: what conditions must be imposed on boundary~$d\varphi$ in order that there exist some inner~$d\varphi$ such that the vector composed of all these (inner and boundary) $d\varphi$ belong to the kernel of~$f_3$? Formulas~\eqref{f3}, together with \eqref{psi0} and~\eqref{psi1}, make it clear that these conditions on $d\varphi_{124},\dots,d\varphi_{235}$ can be written as
\begin{equation}
\left. \begin{array}{r}
d\varphi_{1,124}+d\varphi_{1,143}+d\varphi_{1,135}+d\varphi_{1,152}=0\, , \\
\dotfill \\
d\varphi_{5,512}+d\varphi_{5,523}+d\varphi_{5,531}=0\, .
\end{array} \right\}
\label{*}
\end{equation}
The skipped lines in~\eqref{*} correspond to going around vertices $2$, $3$ and~$4$ along the boundary faces in the same obvious way as the first and the last lines correspond to going around vertices $1$ and~$5$.

It follows from \eqref{phi0} and~\eqref{phi1} that there are $3n$ independent conditions among the $5n$ conditions in~\eqref{*}. Thus, there is also a $3n$-dimensional space of boundary~$d\varphi$ lying in the kernel of~$f_3$ modulo inner~$d\varphi$. Thus, \emph{any of the $6n$ ``boundary'' columns of matrix~$f_3$ is a linear combination of just $3n$ of them modulo ``inner'' columns, and the coefficients in this linear combination are the same for $(f_3^{\rm full})_{\rm l.h.s.}$ and~$(f_3^{\rm full})_{\rm r.h.s.}$}.

This yields immediately the desired proportionality of minors. To calculate the coefficient, it is enough to take any specific minor in~$(f_3^{\rm full})_{\rm l.h.s.}$ and the corresponding minor in~$(f_3^{\rm full})_{\rm r.h.s.}$; this becomes an easy exercise if we take minors containing our $n\times n$ blocks only as a whole.
\end{proof}

\subsection{Tentative invariant}
\label{ss:ti}

It folows from~\eqref{pmatr} that the following function of anticommuting variables at boundary faces is invariant under moves $2\leftrightarrow 3$:
\begin{equation}
\pm \frac{\prod\nolimits^{(\rm f)}\det\zeta_{j_2j_3}}{\prod\nolimits^{(\rm e)}\det\zeta_{i_1i_2}\cdot \prod\nolimits^{(\rm t)}\det\zeta_{k_3k_4}}
\int\!\dots\!\int \prod\nolimits^{(\rm t)} \mathbf f_{k_1k_2k_3k_4} \cdot \prod\nolimits^{(\rm f)} \mathcal Da_{j_1j_2j_3}
\label{ti}
\end{equation}
Here:
\begin{itemize}
	\item the product denoted $\prod\nolimits^{(\rm e)}$ is taken over all inner edges~$i_1i_2$,
	\item both products denoted $\prod\nolimits^{(\rm f)}$ are taken over all inner 2-faces~$j_1j_2j_3$, $j_1<j_2<j_3$,
	\item both products denoted $\prod\nolimits^{(\rm t)}$ are taken over all tetrahedra~$k_1k_2k_3k_4$, $k_1<k_2<k_3<k_4$.
\end{itemize}
The sign ``$\pm$'' in~\eqref{ti} corresponds to the fact that it is a separate problem to order the weights~$\mathbf f$ in their product, as well as the integration in different variables; so we just leave~\eqref{ti} defined up to a sign\footnote{which is quite common when the subject is related to torsions, see Section~\ref{s:inv}}.

\subsection{The case $n=1$: reproducing formula~(\ref{fbf})}
\label{ss:n=1}

Take matrix~\eqref{f3t} with $n=1$ and do the following: multiply the first column, corresponding to $d\varphi_{i_1i_2i_3}$, by~$\zeta_{i_2i_3}$, and similarly the other columns by $\zeta_{i_2i_4}$, $\zeta_{i_3i_4}$ and again~$\zeta_{i_3i_4}$, respectively; then divide the second \emph{row} by~$(-\zeta_{i_3i_4})$. The generating function~\eqref{genfun} for such ``gauge transformed'' matrix is nothing but the ``scalar'' weight~\eqref{fbf}. It is now an easy exercise to deduce equation~\eqref{5f} from the general matrix equation~\eqref{pmatr}.

\section{Arbitrary manifold with one-component boundary: torsion and a set of invariants}
\label{s:inv}

As already mentioned in Subsection~\ref{ss:c}, we are considering a triangulated three-dimensional compact oriented connected manifold~$M$ with \emph{one-component} boundary~$\partial M$. It can be shown that if the triangulation does contain inner vertices, the tentative invariant~\eqref{ti} just turns into zero. Moreover, \eqref{ti} is obviously invariant only with respect to moves $2\leftrightarrow 3$, and nothing is known \textit{a priori} about moves $1\leftrightarrow 4$. This is why we are going to construct in this section the invariants in the case where inner vertices are allowed, and prove their invariants under \emph{all} Pachner moves.

We define the torsion of complex~\eqref{c} as
\begin{equation}
\tau = \frac{\minor f_3}{\minor f_2\, \minor f_4},
\label{tau}
\end{equation}
where the minors correspond to some \emph{nondegenerate $\tau$-chain} according to the usual rules~\cite{turaev}; if such $\tau$-chain does not exist, then $\tau=0$.

\begin{theorem}
\label{th:i}
The expression
\begin{equation}
I_{\mathcal C}(M) = \pm \frac{\prod\nolimits^{(\rm f)}\det\zeta_{j_2j_3}}{\prod\nolimits^{(\rm e)}\det\zeta_{i_1i_2}\cdot \prod\nolimits^{(\rm t)}\det\zeta_{k_3k_4}}\cdot \tau ,
\label{inv}
\end{equation}
where the products are defined in the very same way as explained after formula~\eqref{ti}, taken for given subset~$\mathcal C$ of components of boundary~$d\varphi$ as explained after formula~\eqref{c}, is an invariant of manifold~$M$ with the triangulated one-component boundary~$\partial M$.
\end{theorem}

In other words, $I_{\mathcal C}(M)$ does not change under moves $2\leftrightarrow 3$ and $1\leftrightarrow 4$ within~$M$, not affecting the fixed triangulation of~$\partial M$.

\begin{proof}
A move $2\leftrightarrow 3$ changes only $\minor f_3$ in~\eqref{tau}. The factor by which $\minor f_3$ is multiplied is determined in essentially the same way as in the proof of Theorem~\ref{th:mpen}; the additional tetrahedra (with respect to the siuation where there were just two tetrahedra in the l.h.s.\ and three in the r.h.s.) do not affect this factor.

As for the move $1\to 4$, it can be considered as a composition of moves $0\to 2$ and $2\to 3$, where $0\to 2$ means that we take an inner 2-face let it be face~$123$, and glue in its place two oppositely oriented tetrahedra, say~$1234$, in such way that they are glued to each other by their respective faces $124$, $134$ and~$234$. Thus, old face~$123$ is replaced with a ``triangular pillow'' with the new vertex~$4$ inside.

We thus consider this move $0\to 2$. One possibility of changing the minors in~\eqref{tau} under this move is as follows:
\begin{itemize}
	\item extend $\minor f_2$ by block row corresponding to $d\varphi_{124}$ and block column corresponding to~$dz_4$,
	\item extend $\minor f_3$ by block rows corresponding to $d\psi_{1,1234}^{(1)}$, $d\psi_{2,1234}^{(1)}$ and $d\psi_{1,1234}^{(2)}$, where superscript~$(1)$ indicates one tetrahedron and~$(2)$ --- the other, and block columns corresponding to $d\varphi_{134}$, $d\varphi_{234}$ and one of the two~$d\varphi_{123}$,
	\item extend $\minor f_4$ by block row corresponding to $d\chi_4$ and block column corresponding to~$d\psi_{2,1234}^{(2)}$.
\end{itemize}
Standard argument using block triangularity shows that the three respective minors are thus multiplied by $\frac{\mathcal D\varphi_{124}}{\mathcal Dz_4}$, $\frac{\mathcal D\psi_{1,1234}^{(1)}\wedge \mathcal D\psi_{2,1234}^{(1)}\wedge \mathcal D\psi_{1,1234}^{(2)}}{\mathcal D\varphi_{134}\wedge \mathcal D\varphi_{234}\wedge \mathcal D\varphi_{123}}$, and $\frac{\mathcal D\chi_4}{\mathcal D\psi_{2,1234}^{(2)}}$, where $\mathcal D$ means the exterior product of differentials of $n$ components of the respective quantity. The first of these quantities is computed using \eqref{f2}, the second --- \eqref{f3t}, and the last --- \eqref{f4} together with \eqref{psi0} and~\eqref{psi1}. The result is that $\tau$ is multiplied (up to a sign) by~$\frac{\det\zeta_{14}\,\det\zeta_{34}}{\det\zeta_{23}}$. One can see that, miraculously, this agrees with how the products in formula~\eqref{inv} change.
\end{proof}

Of course, in the case of no inner vertices, the invariants~\eqref{inv} are nothing but coefficients at the products of anticommuting variables in~\eqref{ti}, thus we have proved that these coefficients are topological invariants --- provided a triangulation with no inner points\footnote{and of course not containing edges starting and ending at the same vertex} exists.

\section{Discussion}
\label{s:d}

We have constructed the first ever solution of pentagon equation with anticommuting variables and incorporating, in an essential  way, the noncommutative matrix multiplication; this can be seen in formula~\eqref{f3t} from which the tetrahedron weight is made according to~\eqref{genfun}. We also showed in Section~\ref{s:inv}, on the example of as simple algebraic complexes as we could invent, how the obtained invariants are related to the torsion of acyclic complexes. This also showed the good behavior of our invariants with respect to Pachner moves $1\leftrightarrow 4$ (while the pentagon equation dealt only with moves $2\leftrightarrow 3$).

We plan to write another, and longer, article, containing interesting calculations, especially for ``twisted'' complexes (like those in~\cite{M1,M2}, but for manifolds with boundary), and other material such as the generalization of our complex~\eqref{c} for the case of boundary having any number of components.

Moreover, it turns out that complex~\eqref{c} admits a rather straightforward generalization onto \emph{four}-dimensional manifolds --- this will be the theme of separate research.

The existence of invariants like~\eqref{inv}, with a factor, multiplicative in some values belonging to simplexes of triangulation, multiplied by a Reidemeister-type torsion, always comes as a miracle. The point is that we first construct a complex like~\eqref{c} (already guided by some not very formal ideas), and complex~\eqref{c} belongs to a \emph{fixed} triangulation of a manifold~$M$. It always turns out, however, that the torsion of such a complex behaves beautifully under all types of Pachner moves \emph{changing} the triangulation.

To conclude, we remark that coming from a ``na\"{\i}ve'' state-sum invariant like \eqref{s} or~\eqref{ti} that turns in many cases into zero, to invariants involving torsion can be considered as a sort of renormalization procedure. This procedure introduces new variables $dz_i$ and~$d\chi_i$, and, in physics, such variables may correspond to new physical entities. This raises an interesting question of possible relations between acyclic complexes and renormalization.

\subsection*{Acknowledgements}

One of the authors (I.K.) thanks Irina Aref'eva and all the organizers for the great possibility of making a report at the conference SFT'09, and for their warm hospitality.

The idea of using formulas like \eqref{F1}, \eqref{F2} and~\eqref{F3} was suggested to I.K. by Rinat Kashaev~\cite{kash-pr}. We would like to express him our gratitude for this suggestion.

The work of I.K.\ was supported in part by the Russian Foundation for Basic Research (Grant No. 07-01-00081a).

\end{document}